%% file: main.tex
\documentclass[sigconf,nonacm]{aamas} 

\usepackage{balance} 
\usepackage[utf8]{inputenc}
\usepackage[T1]{fontenc}
\usepackage{multirow}
\usepackage{algorithm}
\usepackage{algpseudocode}
\usepackage{graphicx}
\usepackage{subfigure}
\usepackage{amsmath}
\usepackage{amsthm}

\setcopyright{ifaamas}
\acmConference[AAMAS '24]{Proc.\@ of the 23rd International Conference
on Autonomous Agents and Multiagent Systems (AAMAS 2024)}{May 6 -- 10, 2024}
{Auckland, New Zealand}{N.~Alechina, V.~Dignum, M.~Dastani, J.S.~Sichman (eds.)}
\copyrightyear{2024}
\acmYear{2024}
\acmDOI{}
\acmPrice{}
\acmISBN{}

\acmSubmissionID{269}

\title{Determining Winners in Elections with Absent Votes}

\author{Qishen Han}
\affiliation{
  \institution{Rensselaer Polytechnic Institute}
  \city{Troy}
  \country{United States}}
 \email{hnickc2017@gmail.com}

\author{Am\'elie Marian}
\affiliation{
  \institution{Rutgers University}
  \city{Piscataway}
  \country{United States}}
\email{amelie@cs.rutgers.edu}

\author{Lirong Xia}
\affiliation{
  \institution{Rensselaer Polytechnic Institute}
  \city{Troy}
  \country{United States}}
\email{xialirong@gmail.com}

\begin{abstract}
An important question in elections is the determine whether a candidate can be a winner when some votes are absent. We study this determining winner with  absent votes (WAV) problem when the votes are top-truncated. We show that the WAV problem is NP-complete for single transferable vote, Maximin, and Copeland, and propose a special case of positional scoring rule such that the problem can be computed in polynomial time. Our results in top-truncated rankings differ from the results in full rankings as their hardness results still hold when the number of candidates or the number of missing votes are bounded, while we show that the problem can be solved in polynomial time in either case. 
\end{abstract}


\newcommand{\BibTeX}{\rm B\kern-.05em{\sc i\kern-.025em b}\kern-.08em\TeX}

\begin{document}


\pagestyle{fancy}
\fancyhead{}
\input{macro}


\maketitle 


\section{Introduction} 
In a multi-agent system, voting is one of the most widely applied methods to aggregate preference and make collective decisions. Voting has been rooted in the democratic procedure while emerging as new techniques to other scenarios including search engines~\citep{dwork2001rank}, crowdsourcing~\citep{mao2013better}, and blockchain governance~\citep{grossi2022social}. 

An important question that arises in these scenarios is the need to know the outcome without knowing the preferences of all the voters. 
There are several common reasons in practice why some votes may not be available right away for tallying: delay of absentee ballots, the forecasting of votes with polling results, or the contestation of the validity of some ballots. 

\begin{ex}
    Suppose a city runs its mayoral election and adopts the single transferable vote (STV) (or so-called ranked-choice voting (RCV)) as the voting rule. Unfortunately, some of the absentee ballots have experienced substantial delays and are suspected to be lost. The official investigation will take about one month to locate these missing ballots, causing a significant disruption to the usual political proceedings. Is it possible for the city officials to offer a forecast of all potential winners by considering the current votes and estimating the number of undisclosed ballots? 
\end{ex}

In fact, such examples have happened in practice in U.S. localities that have recently switched to ranked-choice voting. In 2018, the results of the RCV San Francisco mayoral election took a week to be confirmed and tabulated, largely due to the late counting of mail-in ballots.
The results of the 2021 New York City primary election were not certified until {\em a full month} after the election due to the large number of absentee ballots.  These delays, the lack of transparency, and the incomplete information, or lack thereof, on the outcome of cast ballots led to distrust in the election process~\cite{hill2023}

Winner determination with absent votes also provides justification for a candidate's victory in ballots susceptible to manipulation~\citep{baumeister2023complexity}. A proposed heuristic~\cite{jelvaniHCOMP22} empirically evaluated on NYC election night data show promises in identifying election winners, or narrowing down the field of possible winners in a single transferrable vote scenario.

From another perspective, determining winners with absent votes is also related to the classic problem of {\em coalitional manipulation} in computational social choice. In this context, a group of manipulators influence the outcome by strategically adding specific votes to the ballot. 

Originated from the famous Gibbard-Satterthwaite Theorem~\citep{Gibbard73:Manipulation,Satterthwaite75:Strategy}, there have been extensive theoretical studies on this problem in the computational social choice literature from the perspective of coalitional manipulation~\citep{Faliszewski10:AI, Faliszewski10:Using}. Subsequent studies characterize the complexity of such problems for different voting rules including STV~\citep{Bartholdi91:Single}, positional scoring rules~\cite{Davies11:Complexity,Betzler11:Unweighted}, Copeland~\citep{Faliszewski10:Manipulation}, and Maximin~\citep{Xia09:Complexity}. 

However, most previous studies assume that each voter casts a complete linear order, i.e. a {\em full ranking} toward all the candidates. In contrast, votes where voters cast a few {\em top preferences} become increasingly common in real-world scenarios. Top-ranked voting is more practical to implement because it simplifies the computation of the winner and prevents voters without full preferences from casting random votes and corrupting the ballot.
Moreover, the results of coalitional manipulation under full rankings do not extend to the winner prediction for top-ranking votes. \citet{Narodytska14:Computational} studies vote where top-truncated votes are allowed. \citet{menon2015complexity} study a weighted version of coalitional manipulation for top-truncated votes. Yet their hardness results also do not apply due to the unbounded top rankings and the incorporation of weights in the voting process.

Therefore, the following question remains open: {\bf What is the complexity of determining possible winners in top-ranked voting election with absent votes?}










\subsection{Our Contribution}
\begin{table*}[htbp]
\begin{tabular}{@{}llll@{}}
\toprule
 Voting Rule &  Top-$\ell$ & Up-to-$L$  & Full Ranking\\ \midrule
STV  & NPC for $\ell \ge 2$ (Theorem~\ref{thm:STV}) & NPC (Theorem~\ref{thm:stvu}) & NPC~\citep{Bartholdi91:Single} \\
Maximin & NPC for $\ell \ge 2$ (Theorem~\ref{thm:maximin}) & NPC for $L \ge 2$ (Theorem~\ref{thm:maximinu}) & NPC~\citep{Xia09:Complexity}   \\
Copeland  & NPC for $\ell \ge 2$ (Theorem~\ref{thm:copeland}) & NPC for $L \ge 2$ (Theorem~\ref{thm:copelandu})  & NPC~\citep{Faliszewski10:Manipulation}  \\
\multirow{2}{*}{Positional Scoring Rule}  & P for $\ell=2$ (Corollary~\ref{coro:psr2}) & P for $\uparrow$-rounding~\cite{Narodytska14:Computational} & P for plurality and Veto \\
& P for $\scr_2 = \cdots = \scr_{\ell}$ (Theorem~\ref{thm:psr}) & P for $\downarrow$-rounding when $\scr_2 = \cdots = \scr_{L}$ (Theorem~\ref{thm:psrud})   & NPC for Borda~\citep{Davies11:Complexity,Betzler11:Unweighted}  \\ \bottomrule
\end{tabular}
\caption{Complexity of predicting winner with absent votes under full ranking, top $\ell$ ranking, and up to $L$ ranking.\label{tbl:res}}
\end{table*}
We investigate the computational complexity of determining {\em winner with absent votes (WAV)} under multiple voting rules. We focus on two specific settings of top rankings. In the {\em top-$\ell$} setting, every voter is asked to provide their top-$\ell$ preference. And in the {\em up-to-$L$} setting, every voter can list his/her up to $L$ top preference. 

We first show that, when either the number of candidates or the quantity of absent votes is bounded, the WAV problem can be solved in polynomial time under both top-$\ell$ and up-to-$L$ settings. This distinguishes our work from the previous papers under full-ranking settings, in which the hardness results hold even for bounded candidates or a bounded number of manipulators. 

Subsequently, we show that in STV, Maximin, and Copeland, determining the winner with absent votes is NP-complete for every $\ell \ge 2$ in the top-$\ell$ setting and every $L \ge 2$ in the up-to-$L$ setting. Conversely, for the positional scoring rule, we show that the winner can be determined in polynomial time under both settings when the scoring vector does not distinguish the second to the $\ell$-th ($L$-th, respectively) rank. A comparison between our results and the results in full rankings in previous works is in Table~\ref{tbl:res}.

We define the problem in the way of determining the winner with absent votes rather than following the convention of coalitional manipulation because the objects of the two problems are opposite. For the winner determination, we hope the problem is easy so that we have efficient ways to provide possible outcomes to the public. In the coalitional manipulation problem, on the other hand, hardness results are more welcomed because they prevent manipulators from corrupting the elections easily~\citep{Faliszewski10:Using}. Our motivation and inspiration align with the winner determination problem.

\section{Related Works}
As mentioned in the introduction, winner determination with absent votes has been extensively explored by the computational social choice community from the perspective of coalitional manipulation. \citet{Gibbard73:Manipulation} and \citet{Satterthwaite75:Strategy} show that all reasonable voting rules suffer from manipulation under some situations. The earliest studies on the complexity of manipulation problems \citep{Bartholdi89:Computational,Bartholdi91:Single} show that determining even a single manipulator would succeed is NP-hard under some voting rules when the number of candidates is unbounded. A large literature follows the path and develops theoretical results of coalitional manipulation under a spectrum of weighted~\citep{Conitzer02:Complexity,Conitzer07:When,Hemaspaandra07:Dichotomy,Zuckerman09:Algorithms,Xia10:Scheduling} and unweighted~\citep{Xia09:Complexity,Betzler11:Unweighted,Davies11:Complexity,Narodytska11:Manipulation,Faliszewski10:Manipulation} voting rules. However, most of the previous work focuses on featuring full rankings. \citet{menon2015complexity} studies the complexity of weighted coalitional manipulation when top-truncated votes are allowed. However, the incorporation of the weights prevents its results from extending to the unweighted version. 

A closely related problem to the winner with absent votes and coalitional manipulation is the {\em possible winner} problem. The problem takes a set of candidates and a profile of partial orders on the candidates and asks if there is a full-order profile that extends the partial orders and makes a certain candidate a winner. A coalitional manipulation instance can be seen as a possible winner instance where a portion of the profile is full orders and the rest is empty votes. When the number of candidates is bounded, the possible winner problem can be solved in polynomial time in unweighted votes and NP-complete in weighted votes under multiple rules~\citep{Conitzer07:When,Pini11:Incompleteness,Walsh07:Uncertainty}. When the number of candidates is unbounded, the problem is P in the Condorcet rule~\citep{Konczak05:Voting} yet NP-complete in a large variety of other rules~\citep{Bartholdi91:Single,Xia08:Determining,Betzler10:Towards}.

Recent work has also looked at the intersection of voting theory with regulatory frameworks in the context of ranked-choice voting elections. In particular, there has been an interest in defining and computing the margin of victory (MoV), an important robustness measure of elections in Australia, where small margins would trigger elections audits~\cite{blom2016efficient,magrino2011computing}, or
potentially result in shifts in the balance of power~\cite{blom2020power,blom2020lostballot}, and in election manipulation~\cite{blom2019election} in ranked-choice voting settings.

\section{Preliminaries}
Let $M$ be the set of {\em candidates} (or {\em alternatives}). Let $m = |M|$ denote the number of candidates. Given a positive integer $\ell$, a top $\ell$ ranking $R$ is a ranking on a $\ell$-subset of $M$, where all the unranked candidates are tied and ranked lower than the ranked candidates.  
Let $\mathcal{L}_{\ell}(M)$ denote the set of all top $\ell$ rankings (or linear orders) on $M$. 

There are in total $\vt + \un$ voters in the vote, where $\vt$ is the number of voters whose votes are known, and $\un$ is the number of voters whose votes are absent. In the {\em top-$\ell$ setting}, each voter casts a top $\ell$ ranking $R\in \mathcal{L}_{\ell}(M)$ to represent their preference, where $a\succ_{R} b$ means the voter prefers $a$ to $b$. In the {\em up-to-$L$ setting}, each agent cast a ranking $R \in (\bigcup_{i=1}^{\ell} \mathcal{L}_{i}(M))$.  The vector of all voters' votes is called a {\em profile}. 
Let $P$ denote the profile of known votes and $P'$ denote the profile of absent votes.

Given a profile $P$ and two alternatives $a$ and $b$, let $P[a\succ b]$ denotes the votes in $P$ that prefer $a$ to $b$. The weighted majority graph (WMG) of $P$ is a graph whose vertices are votes in $P$ and weights on $a \to b$ is $ \omega_P(a\to b) = P[a\succ b] - P[b\succ a]$. 

\begin{ex}
\label{ex:votes}
    Let $M = \{1, 2, 3, 4\}$, $n = 4$, and $t = 2$.

    Let $P_1$ be a profile of 4 votes under the top-$2$ setting. $P_1$ contains two votes for $[3\succ 1]$, one vote for $[1\succ 4]$, and one vote for $[2\succ 1]$. 

    Let $P_2$ be a profile of 4 votes under the up-to-$3$ setting. $P_1$ contain two votes for $[1\succ 3]$, one vote for $[2\succ 1\succ 4]$, and one vote for $[3]$. 

    The weighted majority graph of $P_1$ and $P_2$ is in Figure~\ref{fig:ex_votes}.

    \begin{figure}
        \centering
        \includegraphics[width=0.9\columnwidth]{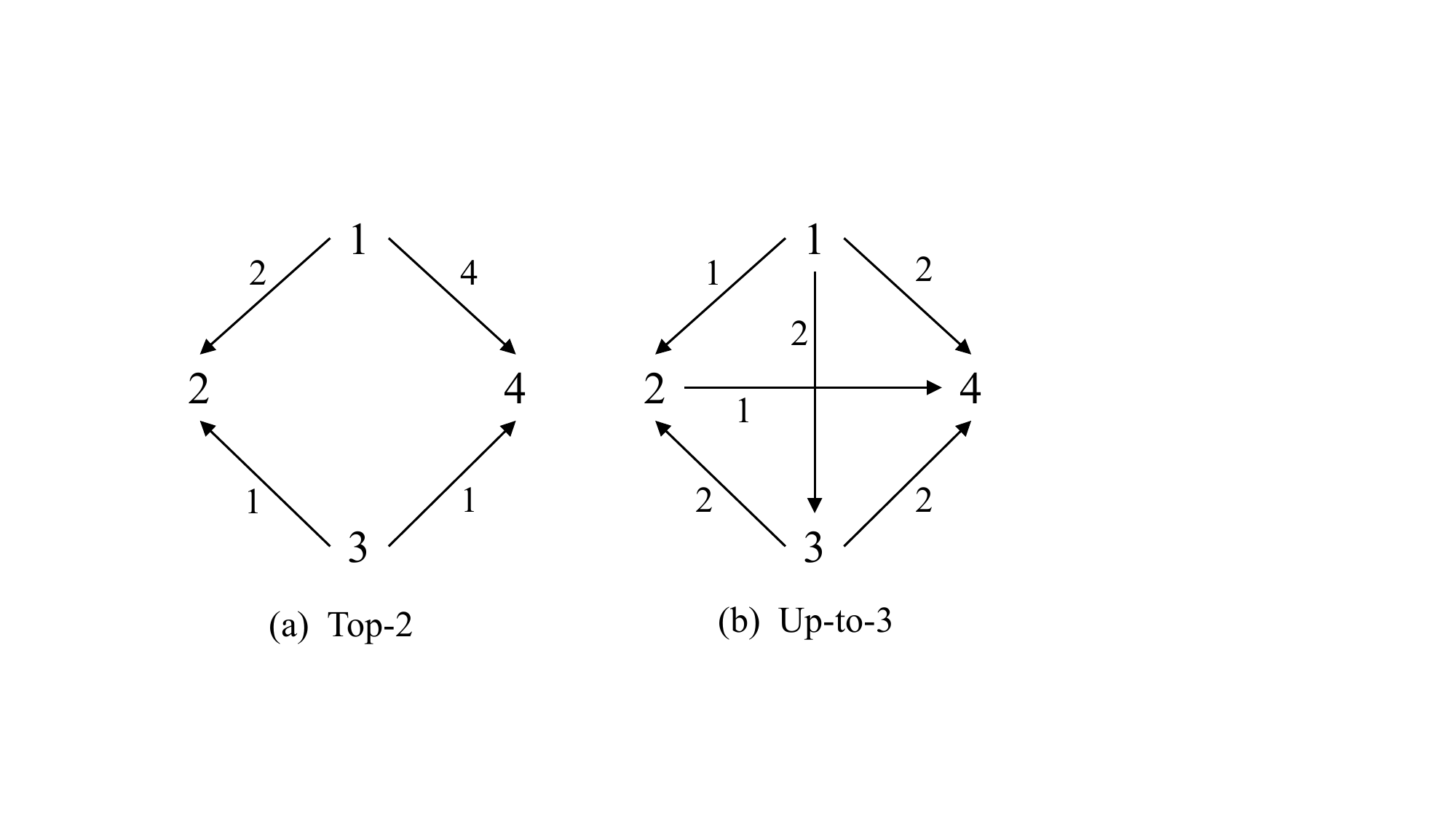}
        \caption{Weighted majority graph of instances in Example~\ref{ex:votes}.}
        \label{fig:ex_votes}
    \end{figure}
\end{ex}

\subsection{Voting Rules}
A (resolute) voting rule takes a voting profile as input and outputs a unique candidate as the winner. In the top-$\ell$ setting, a voting rule $r_{\ell}: (\mathcal{L}_{\ell}(M))^* \to M$, and in the up-to-$L$ setting, a voting rule $\overline{r}_L: (\bigcup_{i=1}^{L} \mathcal{L}_{i}(M))^* \to M$. A voting rule is {\em anonymous} if the winner is insensitive to the identities of agents.

We focus on the variation of the following common voting rules for the top $\ell$ or up to $L$ rankings. For a voting rule $r$. we use $r_{\ell}$ to denote its variation in top-$\ell$ setting and $\overline{r}_L$ to denote its variation in up-to-$L$ setting.

 The {\em single transferable voting} (STV) elects the winner in at most $m-1$ rounds. In each round, each vote contributes 1 score to its most preferred candidate that has not been eliminated, and the candidate with the lowest score is eliminated in that round. A tie-breaking mechanism is applied to select a single loser when necessary. If all candidates ranked in a vote are eliminated, that vote does not contribute to any candidate. The candidate remaining to the last becomes the winner. 

The {\em Copeland} rule is parametrized by a real number $0\le \alpha \le 1$, denoted by $\cpd^{\alpha}$. Given a profile $P$, a candidate $a$ gains 1 score for every other candidate $b$ it beats in the head-to-head competition (the weight on edge $a \to b$ is positive in the WMG) and $\alpha$ score when there is a tie. The candidate with the highest Copeland score becomes the winner, and a tie-breaking mechanism is applied to select a single winner if necessary. 

 In the {\em Maximin} rule, the {\em min-score} of a candidate $a$ is the lowest weight of its out-going edges in the weighted majority graph, i.e. $\min_{b\in M\setminus \{a\}} \omega_P(a\to b)$. The candidate with the highest min-score becomes the winner, and a tie-breaking mechanism is applied to select a single winner if necessary.
    
 We define the {\em (integer) positional scoring rule} in two settings respectively. 

    In the top-$\ell$ setting, a positional scoring rule is characterized by an $\ell$-dimension vector $\vs_{\ell} = (\scr_1, \scr_2,\cdots, \scr_{\ell})$ with $\scr_1 \ge \scr_2 \ge \cdots \ge \scr_{\ell} \ge 0$. Given a top-$\ell$ vote $V_i$ and a candidate $\tg$, let $s(V_i, \tg) = \scr_j$ where $j$ is the rank of $\tg$ in $V_i$ or $s(V_i, \tg) = 0$ if $\tg$ is not ranked in $i$. For any profile $P$, let $s(P, \tg) = \sum_{V_i\in P} s(V_i, \tg)$. The candidate $\tg$ maximizing $s(P, \tg)$ becomes the winner, and a tie-breaking mechanism is applied to select a single winner if necessary.

    In the up-to-$L$ setting, we follow the scheme from~\citet{Narodytska14:Computational} to deal with top-truncated rankings. A positional scoring rule is characterized by an $L$-dimensional vector, $\vs_{L} = (\scr_1, \scr_2,\cdots, \scr_{L})$ with $\scr_1 \ge \scr_2 \ge \cdots \ge \scr_{L} \ge 0$, and a {\em rounding indicator}, denoted by $\uparrow$ or $\downarrow$. In an {\em up-rounding} scoring rule $\vs_{L\uparrow} $, a candidate $\tg$ ranked $j$-th in an $\ell$-ranking vote $V_i$ has a score $s(V_i, \tg) = \scr_j$. And in a {\em down-rounding} scoring rule $\vs_{L\downarrow} $, a candidate $\tg$ ranked $j$-th in an $\ell$-ranking vote $V_i$ has a score $s(V_i, \tg) = \scr_{L - \ell + j}$. In both cases, an agent not ranked in a vote gets a score of 0. For any profile $P$, let $s(P, a) = \sum_{V_i\in P} s(V_i, \tg)$. The candidate $\tg$ maximizing $s(P, \tg)$ becomes the winner, and a tie-breaking mechanism is applied to select a single winner if necessary.

\begin{ex}
    \label{ex:STV}
    We calculate the STV winner for the top-$\ell$ instance in Example~\ref{ex:votes}. 
    In the first round, candidate $3$ gets two votes, $1$ and $2$ get one vote each, and $4$ gets no votes. Therefore, $4$ is eliminated. 

    In the second round, candidate $3$ gets two votes, and $1$ and $2$ get one vote each. Suppose we use a lexicographic tie-breaking mechanism. Therefore, $2$ is eliminated. 

    In the third round, the vote $[2\succ 1]$ contributes to candidate $1$. Therefore, both $1$ and $3$ get two votes. Then $3$ is eliminated, and $1$ becomes the winner. 
\end{ex}

\begin{ex}
    \label{ex:PSR}
    We calculate the winner for the up-to-$L$ instance in Example~\ref{ex:votes} under up-rounding and down-rounding positional scoring rules respectively. The scoring vector $\vs_L = (8, 2, 1)$. 

    In the up-rounding rule, the score of $1$ is $8\times 2 + 2 = 18$, of $2$ is $8$. of $3$ is $8 + 2 = 10$, and of $4$ is 1. Therefore, $1$ is the winner. 

    In the down-rounding rule, the score of $1$ is $2\times 2 + 2 = 6$, of $2$ is $8$, of $3$ is $1\times 2 + 8 =10$, and for $4$ is $1$. Therefore, $3$ is the winner.  
\end{ex}

\subsection{Computational Problems}
We consider the following computational question: when $\ell$ (or $L$, respectively) is a constant, given a set of candidates $M$, a set of known profiles $P$ of top-$\ell$ (up-to-$L$, respectively) votes, the number of absent votes $t$, and a targeted candidate $\tg$, is there a profile $P'$ of $t$ votes that makes $\tg$ the winner? 

We first define the question for the top-$\ell$ setting. For each constant $\ell \ge 1$ and a voting rule for top-$\ell$ rankings, we define the following problem. 

\begin{dfn}[\wpp-$r_{\ell}$]
    \label{dfn:wppt}
    \textbf{Input}: a set of candidates $M$, a profile $P$ of known top-$\ell$ ranking votes, the number of absent votes $\un$, and a candidate $\tg$. 
    
    \textbf{Determine}: If there exists a profile $P'$ of $\un$ top-$\ell$ ranking votes such that $r_{\ell}(P\cup P') = \tg$.
\end{dfn}

We also consider two variations of the \wpp{} problem with fixed parameters. In \wpp{} with fixed $m$, the number of the candidates is removed from the input and becomes a pre-determined constant. In \wpp{} with fixed $\un$, the quantity of the absent votes becomes a pre-determined constant. 

    
    

For the up-to-$L$ setting, we follow a similar definition. For each constant $L \ge 1$ and a voting rule $\overline{r}_L$, we define the following problems. 

\begin{dfn}[\wpp-$\overline{r}_L$]
    \label{dfn:wppu}
    \textbf{Input}: a set of candidates $M$, a profile $P$ of known up-to-$L$ ranking votes, the number of absent votes $\un$, and a candidate $\tg$. 
    
    \textbf{Determine}: If there exists a profile $P'$ of $\un$ up-to-$L$ ranking votes ranking votes such that $\overline{r}_L(P\cup P') = \tg$.
\end{dfn}

Similarly, we also consider \wpp{} with fixed $m$ and \wpp{} with fixed $\un$ in the up-to-$L$ setting. 

For both problems, we focus on $\ell \ge 2$ and $L \ge 2$ cases, as when $\ell = 1$ or $L = 1$, most common voting rules reduce into plurality, and the WAV problem can be computed in polynomial time. 

    


    

\section{Fixed $m$ and fixed $t$}
We first show the easiness result for the variation of a fixed number of candidates and a fixed number of absent votes under both settings. 

\begin{prop}
    \label{prop:fixt}
    For any $\ell \ge 2$ and any anonymous voting rule $r_{\ell}$, both \wpp-$r_{\ell}$ with any fixed $m \ge 2$ and \wpp-$r_{\ell}$ with any fixed $\un$ can be solved in polynomial time if the winner of $r_{\ell}$ can be computed in polynomial time. 
\end{prop}

\begin{proof}[Proof Sketch]
    \noindent\textbf{Fixed $m$.} We enumerate all possible {\em anonymous} profiles $P'$ of $t$ votes. There are $\frac{m!}{(m -\ell) !} = O(m^{\ell}) = O(1)$ different top-$\ell$ rankings. The numbers of these rankings sum up to $t$. Therefore, there will be at most $O(t^{m^{\ell}}) = poly(t)$ many anonymous profiles. 
    
    \noindent\textbf{Fixed $\un$.} We enumerate all possible profiles $P'$ of $t$ votes. For each vote, there are $O(m^{\ell})$ different top-$\ell$ rankings. Therefore, the number of all possible $P'$ is at most $O(m^{t\ell})$. 
\end{proof}

\begin{prop}
\label{prop:fixu}
    For any $L \ge 2$ and any anonymous voting rule $\overline{r}_{L}$, both \wpp-$\overline{r}_{L}$ with any fixed $m \ge 2$ and \wpp-$\overline{r}_{L}$ with any fixed $\un$ can be solved in polynomial time if the winner of $\overline{r}_L$ can be computed in polynomial time. 
\end{prop}

Proposition~\ref{prop:fixt} and~\ref{prop:fixu} directly imply that previous results for full-rankings votes do not apply to our top-$\ell$ and up-to-$L$ settings, as their hardness results hold even for fixed $m$ or fixed $t$ instances.
In the rest of the paper, we focus on the problem with variables $t$ and $m$ under common voting rules.







\section{Single Transferable Vote} 
For STV, we show that \wpp-${\stv_{\ell}}$ is NP-complete for all $\ell \ge 2$. 

\begin{thm}
\label{thm:STV}
For every constant $\ell\ge 2$, \wpp-${\stv_{\ell}}$ is NP-complete.
\end{thm}
\begin{proof}
The membership of NP is straightforward by running the vote and checking the winner. The hardness is proved by a reduction from {\em restricted exact three-cover (RXC3)} that is similar to the reduction in the hardness proof for the manipulation problem under STV~\citep{Bartholdi91:Single}. We first show the case of $\ell\ge 4$, and describe on how to modify to $\ell=2$ and $\ell=3$ cases. 
\begin{dfn}[RXC3~\cite{Gonzalez1985:Clustering}]
    \label{dfn:rxc} {\bf Input}: (1) a set of $q$-elements, denoted by $\res = \{\re_1, \re_2,\ldots, \re_q\}$, where $q$ is divisible by $3$; (2) $q$ sets $\rss = \{\rs_1, \rs_2,\ldots, \rs_q\}$ such that for every $j\le q$, $\rs_j\subseteq \res$ and $|\rs_j|=3$. For every $i\le q$, $\re_i$ is in exactly three sets in $\rss$. Without loss of generality, we assume that $q$ is an even number. If $q$ is odd, then we use an instance with duplicate $\res$ and $\rss$. 
    
    {\bf Determine}: if there exists a subset $\rss^* \subseteq \rss$ such that for every $\re_i \in \res$, there exists exactly one $\rs_j \in \rss^*$ such that $\re_i \in \rs_j$. We call $\rss^*$ an {\em exact 3-cover} of $\res$. Note that if such $\rss^*$, there must be $|\rss^*| = \frac{q}{3}$. 
\end{dfn}

For any RXC3 instance $(\res, \rss)$, where $\res = \{\re_1, \re_2,\ldots, \re_q\}$ and $\rss = \{\rs_1, \rs_2,\ldots, \rs_q\}$. We construct the following \wpp-${\stv_{\ell}}$ instance. 

{\bf Candidates:} there are $3q+3$ alternatives $\{\wn,\tg\}\cup \{\da_0,\da_1,\ldots,\da_q\}\cup \{\ba_1,\bb_1,\ldots,\ba_q,\bb_q\}$. We assume that $\da_0\succ\da_1\succ \da_2\succ\cdots\succ \da_q \succ \ba_1\succ\bb_1\succ\ba_2\succ\bb_2\succ\cdots\succ\ba_q\succ\bb_q$ in tie-breaking.

{\bf Absent Votes:} $t = q/3$. 

{\bf Known votes:} The profile $\prfn$  consists of the following votes, where only the top preferences are specified. We'll show that either $\wn$ or $\tg$ is the winner, therefore, the votes can be filled to top-$\ell$ ranking arbitrarily without affecting the proof. 
\begin{itemize}
\item {\boldmath $\prf_1$}: There are $12q$ votes of $[\tg\succ \wn]$.
\item {\boldmath $\prf_2$}: There are $12q-1$ votes of $[\wn\succ \tg]$.
\item {\boldmath $\prf_3$}: There are $10q+ 2q/3$ votes of $[\da_0\succ \wn\succ \tg]$.
\item {\boldmath $\prf_4$}: For every $i\in \{1,\ldots,q\}$, there are  $12q-2$ votes of $[\da_i\succ \wn\succ \tg]$.
\item {\boldmath $\prf_5^1$}: For every $i\in \{1,\ldots,q\}$, there are $6q+4i-6$ votes of $[\ba_i\succ\bb_i\succ \wn\succ \tg]$; and {\boldmath $\prf_5^2$}: for every $i\in \{1,\ldots,q\}$ and  every $j$ such that $\re_j\in \rs_i$, there are two votes of $[\ba_i\succ \da_j\succ \wn\succ\tg]$. 
\item {\boldmath $\prf_6^1$}:  For every $i\in \{1,\ldots,q\}$, there are $6q+4i-2$ votes of $[\bb_i\succ \ba_i\succ \wn\succ \tg]$; and {\boldmath $\prf_6^2$}: for every $i\in \{1,\ldots,q\}$, there are two votes of $[\bb_i\succ \da_0\succ \wn\succ \tg]$.
\end{itemize}

First, we show that no matter what rankings $\prft$ contains, {\bf the winner will be either $\tg$ or $\wn$}. This is because, once one of $\tg$ or $\wn$ is eliminated in some round, the remaining other will get its votes of $\prf_1$ or $\prf_2$: if $c$ is eliminated first, $\wn$ will get $12q$ votes from $\prf_1$; and if $\wn$ is eliminated first $\wn$ will get $12q-1$ votes from $\prf_2$. Therefore, the remaining one will have a score of at least $24q-1$. On the other hand, all other alternatives cannot have such a high score:
\begin{itemize}
    \item $\ba_i$ and $\bb_i$ only gets votes from each other and from $\prft$, and cannot have score more than $12q + 8i+ q/3 \le 20q + q/3$. 
    \item $\da_0$ can get two votes from each $\bb_i$ and votes from $\prft$, and cannot have score more than $10q + 2/3q + 2q + q/3 = 13q$. 
    \item $\da_i$ can get votes from $\ba_i$ and from $\prft$. Since each $\re_j$ is in exactly three $\rs_i$, $\da_i$ can get two votes from exactly three $\ba_i$. Therefore, $\da_i$ cannot have score more than $12q - 2 + 6 + q/3 = 12q + q/3 +4$. 
\end{itemize}
Therefore, no other alternative can exceed the score of $24q-1$ at any time, and cannot be the winner. 

Now we show that \wpp-${\stv_{\ell}}$ is a YES instance if and only if RXC3 is a YES instance. 

{\bf Suppose RXC3 is a YES instance.} $\rss^*$ be an exact 3-cover of $X$, and $I = \{i \mid \rs_i \in \rss^*\}$ be the index set of $\rss^*$. Then we construct $\prft$ as follows: for each $i\in \idx$, there is one vote of $[\bb_i\succ\ba_i\succ \tg\succ \wn]$. In the first $q$ round of voting, for each $i\le q$, if $i\in\idx$, $\ba_i$ is eliminated; otherwise $\bb_i$ is eliminated. If $\ba_i$ is eliminated, $6q+4i-6$ of its vote transfer to $\bb_i$, and for every $j$ such that $\re_j\in\rs_i$, $\da_j$ gets two of its vote. If $\bb_i$ is eliminated, $6q+4i-2$ of its vote transfer to $\ba_i$, and two transfer to $\da_0$. Therefore, in the beginning of $q+1$ round, the plurality scores of the remaining alternatives are as in the following table.
\renewcommand{\arraystretch}{1.5}
\begin{center}
\begin{tabular}{|@{\ }c|@{\ }c@{\ } @{\ }|@{\ }c@{\ }|@{\ }c@{\ }|@{\ }c@{\ }|@{\ }c@{\ }|}
\hline  
Rd. & $\wn$ & $\tg$ & $\ba_i$ or $\bb_i$& $\da_0$& $ \da_i$\\
\hline $q+1$& $12q-1$ & $12q$ &  \begin{tabular}{@{}c@{}}$12q+8i-1$ or  \\$12q+8i-5$\end{tabular}& $12q $ & $12q $\\
\hline
\end{tabular}
\end{center}
\renewcommand{\arraystretch}{1}
Therefore, $\wn$ is eliminated in round $q+1$, and $\tg$ will become the winner eventually. 

{\bf Suppose \wpp-${\stv_{\ell}}$ is a YES instance.} We prove that RXC3 is a YES instance in the following steps. 

{\bf Step 1.} In the first $q$ rounds, exactly one of $\ba_i$ and $\bb_i$ is eliminated for all $i\le q$. Firstly, the initial score of $\ba_i$ and $\bb_i$ is at most $6q+4i+q/3 \le 10q+q/3$, while the score of other alternatives is at least $10q + 2/3q$. On the other hand, once one of $\ba_i$ and $\bb_i$ is eliminated, the other gets the transferred votes and has a score of more than $12q$. Therefore, in the first $q$ round, in each round, either $\ba_i$ or $\bb_i$ is eliminated for a distinct $i$. 

{\bf Step 2.} Let $\idx = \{i:\ \ba_i\text{ is eliminated in the first $q$ rounds.}\}$. Then $\idx$ must be the index set of an RXC3 solution.  

Suppose it is not the case. Firstly, there must be $|\idx| \le q/3$. For each $i \in \idx$, $\bb_i$ needs at least one vote for $P'$ to win $\ba_i$ in the round that $\ba_i$ is eliminated. And once $\bb_i$ is not eliminated, this vote follows $\bb_i$ to the round $q+1$ and cannot contribute to another $\bb_i$'s winning. Therefore, since $t = q/3$, there are at most $q/3$ of $\bb_i$ that beats $\ba_i$. 

Now suppose $|\idx| = q/3$ and $\idx$ is not a solution. Then there exists some $\re_j \in \res$ that is not covered. In this case, for all $i$ such that $\re_j\in\rs_i$, $\bb_i$ is eliminated. Therefore, $\da_j$ does not get any vote transfer from the first $q$ rounds. On the other hand, all $q/3$ votes contribute to some $\bb_i$ in the $q+1$ round and cannot contribute to $\da_j$. Therefore, $\da_j$ has a score of $12q - 2$ in the $q+1$ round and is eliminated, and its votes are transferred to $\wn$. Then $\wn$ has a score of at least $24q -3$ which exceeds $\tg$ all the time. Therefore, $\tg$ cannot be the winner.

Finally, suppose $|\idx| < q/3$. Then there are at least $q - 3|\idx|$ of $\re_j$ not covered, and all the corresponding $\da_j$ do not get transferred in the first $q$ rounds. If they shall not be eliminated in the $q+1$ round, there needs to be at least 1 vote for each of them, which is $q - 3|\idx|$ votes from $P'$. However, since all $|\idx|$ votes contribute to some $\bb_i$, there are only $q/3 - |\idx|$ and cannot cover all $\da_j$. Therefore, one $\da_j$ is eliminated in this round, and its votes are transferred to $\wn$. Then $\wn$ has a score of at least $24q -3$ which exceeds $\tg$ all the time. Therefore, $\tg$ cannot be the winner.

Therefore, once \wpp-${\stv_{\ell}}$ is a YES instance, the index set of eliminated $\ba_i$ in the first $q$ rounds forms the index set of a solution to the RXC3. Therefore, RXC3 is a YES instance. 

\paragraph{$\ell=2$ and $\ell=3$} To prove the $\ell=2$ and $\ell=3$ case, all we need to do is to truncate the rankings in $\ell=4$ construction into top-2 or top-3 rankings respectively. The rest of the proof will remain the same (except for a few vote transition which does not affect the winner). The full proof is in the appendix.
\end{proof}

\begin{thm}
    \label{thm:stvu}
    For every constant $L \ge 2$, \wpp-$\overline{STV}_{L}$ is NP-complete. 
\end{thm}
The proof for the up-to-$L$ case follows the top-$\ell$ case by replacing all $\ell$ to $L$. The construction of the \wpp{} instance requires $P'$ to make use of all $L$ positions in every vote to make $\tg$ the winner. 

\section{Maximin}
For Maximin and Copeland, we leverage the following lemma to construct the instance in the reduction. 

\begin{lem}
    \label{lem:McGarvey}
    For any constant $\ell \ge 2$, an arbitrary set of candidates $M$ with $m \ge \ell$, and two arbitrary candidates $a, b \in M$, there exists a voting profile $P$ with $poly(m)$ of top-$\ell$ ranking votes, and the weighted majority graph of $P$ contains only one non-zero-weighted edge of $a\to b$ with weight 2. 
\end{lem}

Lemma~\ref{lem:McGarvey} enables us to add arbitrary edges with even number weights to a weighted majority graph in polynomial many votes. 

\begin{proof}
   Our construction of $P$ follows the spirit of McGarvey~\cite{McGarvey53:Theorem}. It takes two steps: 

   \noindent\textbf{Step 1:} We first construct a slightly different profile $P'$. For any $\ell$-subset $M_{\ell}$ of $M$, and any permutation $\sigma_{M_{\ell}}$ on $M_{\ell}$, $P'$ contains a vote for $[\sigma_{M_{\ell}}(1) \succ \sigma_{M_{\ell}}(2) \succ \cdots \succ \sigma_{M_{\ell}}(\ell)]$. The number of votes in $P'$ is $A_{\ell}^{m} = O(m^{\ell})$. Due to symmetricity, all the candidates are tied in $P'$, and the weights of all the edges are 0 in the WMG of $P'$. 

   \noindent\textbf{Step 2:} Pick one vote in $P'$ such that $b$ is ranked the top and $a$ is ranked the second. Then $P$ is constructed by swapping $a$ and $b$ in this vote while keeping all other votes unchanged in $P'$. Since the only change is the relative position between $a$ and $b$ in one vote, the WMG of $P$ contains only one edge which is $a\to b$ with weight 2. And $P$ also contains $O(m^{\ell})$ edges. 
\end{proof}

We show that for $\maximin_{\ell}$ it is NP-complete to determine whether a candidate is a possible winner for all $\ell\ge 2$. 

\begin{thm}
\label{thm:maximin}
    For all constant $\ell \ge 2$, \wpp-${\maximin_{\ell}}$ is NP-complete. 
\end{thm}

\begin{proof}
The membership of NP is held by running the vote and checking the winner. For the hardness, we give a reduction from RXC3.

For any RXC3 instance $(\res, \rss)$, where $\res = \{\re_1, \re_2,\ldots, \re_q\}$ and $\rss = \{\rs_1, \rs_2,\ldots, \rs_q\}$, we construct the following \wpp-${\maximin_{\ell}}$ instance.

\noindent\textbf{Candidates.} There are $2q+\ell$ candidates: $\res\cup \rss\cup\{\tg\} \cup W$, where $W = \{ w_1, w_2, \cdots, w_{\ell-1}\}$. We assume $\tg\succ \re_1\succ \re_2\succ \cdots \succ \re_q$ in tie-breaking. 

\noindent\textbf{Absent votes.} $t = \frac{q}{3(\ell -1)}$. (Without loss of generality, we assume that $q$ can be divided by $3(\ell -1)$. With not, we duplicate both $\res$ and $\rss$ for $3(\ell -1)$ times.)

Before presenting the profile, we give an intuition of our construction. The min-score of each $\re_i$ comes from three $\rs_j \ni \re_i$, and is exactly $\frac{q}{3(\ell-1)} + 1$ higher than $\tg$'s min-score. To make $\tg$ the winner, $\tg$ appears in all votes in $P'$ to increase its min-score by $\frac{q}{3(\ell-1)}$. Moreover, the $\rs_j$ that appears in $P'$ should consist of an exact 3-cover so that the min score of every $\re_i$ decreases by $1$. If not, $\tg$ will be beaten by some $\re_i$ not covered. 

\noindent\textbf{Known votes.} We construct most of the profile via Lemma~\ref{lem:McGarvey} as follows. 
\begin{itemize}
    \item For each $\re_i\in \res$ and $\rs_j \in \rss$ such that $\re_i\in \rs_j$, there is an edge $\rs_j \to \re_i$ with weight $q$. 
    \item For each $\re_i \in \res$ there is an edge $\re_i \to \tg$ with weight $q + \frac{q}{3(\ell -1)} + 2$; for each $\rs_j \in \rss$ there is an edge $\rs_j \to \tg$ with weight $q + \frac{q}{3(\ell -1)} + 2$. 
    \item For each $\wn_i\in W$, there is an edge $\re_j\to\wn_i$ for each $\re_i \in \res$ with weight $2q$ and an edge $\rs_j \to \wn_i$ for each $\rs_j \in \rss$ with weight $2q$.  
    \item For each $\wn_i\in W$, there is an edge $\tg\to\wn_i$ with weight $2q$. 
\end{itemize}
Additionally, there is one vote for $[\tg \succ \wn_1 \succ \cdots \succ \wn_{\ell - 1}]$. The WMG of $P$ is in Figure~\ref{fig:maximin}. 

\begin{figure}
    \centering
    \includegraphics[width=\columnwidth]{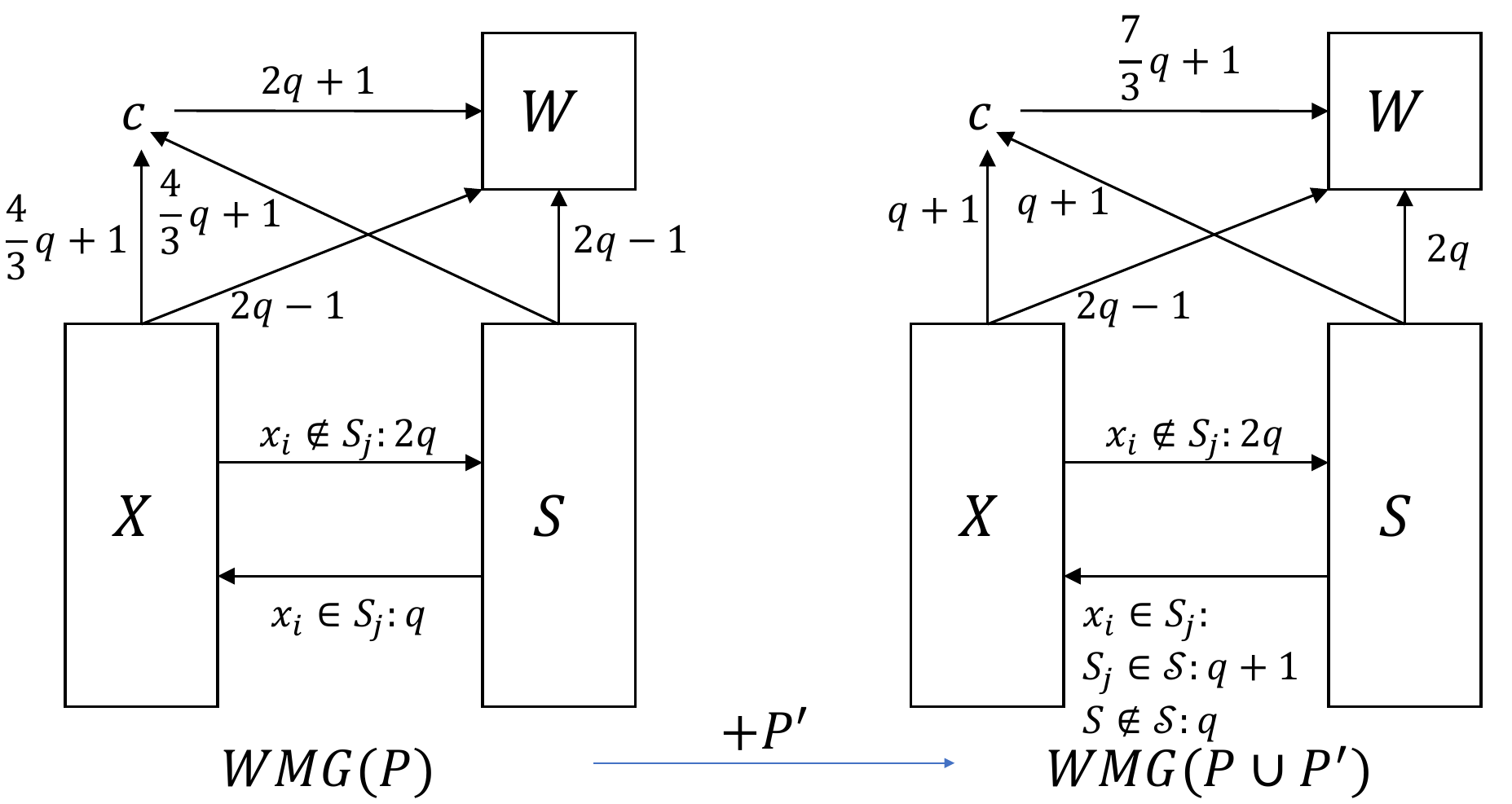}
    \caption{Weighted majority graph for Maximin}
    \label{fig:maximin}
\end{figure}

Then in profile $P$, the min score of $\tg$ is $-q - \frac{q}{3(\ell -1)}-1$, of each $\wn_i$ is $-2q-1$, of each $\re_i$ is $-q$ (from $\rs_j \ni \re_i$), and for each $\rs_j$ is $-2q$.

\paragraph{Suppose RXC3 is a YES instance.} And let $\rss^*$ be the exact 3-cover of $\res$. Then we construct $P'$ as follows: $\tg$ is ranked the top and followed by $\ell -1$ of $\rs_j$ in each vote. The set of all the $\rs_j$ ranked the second to the $\ell$-th in $P'$ (which is exactly $\frac{q}{3(l-1)} \times (l-1) = \frac{q}{3}$ of $\rs_j$) is exactly $\rss^*$.  Now we show that $\maximin_{\ell}(P\cup P') = \tg$. In $P\cup P'$, the min-score of $\tg$ is $-q-1$. For any $\re_i$, since $\rss^*$ is an exact 3-cover, there exists a $\rs_j\in \rss^*$ such that $\re_i \in \rs_j$. Since there is one vote that ranked $\rs_j$ higher than $\re_i$ in $P'$, the min-score of $\re_i$ in $P\cup P'$ is $-q -1$. The min-score of any $\rs_j$ or any $\wn_i$ will not exceed $-2q + 1$. Therefore, $\tg$ becomes the winner. 

\paragraph{Suppose \wpp-${\maximin_{\ell}}$ is a YES instance.} And let $P'$ be a profile of $\un$ votes such that $\maximin_{\ell}(P\cup P') = \tg$. We proceed with the proof in two steps. 

\noindent\textbf{Step 1.} Without loss of generality, we can assume that $\tg$ is ranked the first in all votes in $P'$.  Suppose there is a vote $V = [\tg_1 \succ \tg_2\succ \cdots, \tg_{\ell}]$ in $P'$ such that $\tg_i \in \res\cup\rss\cup\{c\}\cup W$ for $i = 1,2\cdots, \ell$, and $\tg_1 \neq \tg$. Then we construct a new vote $V'$ by the following rule:
\begin{itemize}
    \item If there is some $1 < i < \ell$ such that $\tg_i = \tg$, then $V' = [\tg \succ \tg_1 \succ \tg_{i-1}\succ \tg_{i+1} \succ \cdots \succ \tg_{\ell}]$. 
    \item If $\tg_{\ell} = \tg$ or $\tg$ is not ranked in $V$, then $V' = [\tg \succ \tg_1 \succ \cdots \succ \tg_{\ell -1}]$.
\end{itemize}
We create a new profile $P''$ such that $V$ is substituted by $V'$ while all other votes remain the same as in $P$. Firstly, since $\tg$ becomes the first rank, the min-score of $\tg$ will not decrease by this substitution. On the other hand, the order of other candidates is unchanged in $V'$, the min-score of any other candidate will not increase. Therefore, $P''$ is still a solution that makes $\tg$ a winner. 

\noindent\textbf{Step 2.} For every vote in $P'$, the second to $\ell$-th rank is some $\rs_j \in \rss$, and the set of all these $\rs_j$ (denoted by $\rss^*$) is an exact 3-cover of $\res$. 
Since $\tg$ is ranked top in all votes in $P'$, we know that the min score of $\tg$ in $P\cup P'$ is exactly $-q-1$. To make $\tg$ the winner, the min-score of any $\res_i$ should not exceed $-q-1$. Let $\rss^*$ be the multiset of all $\rs_j$ ranked by some vote in $P'$. Suppose $\rss^*$ is not an exact 3-cover of $\res$. Since there are at most $\frac{q}{3}$ positions for $\rs_j$, there must exist an $\re_i$ not been covered by $\rss$. Then we consider the min-score of $\re_i$. 
\begin{itemize}
    \item $\rs_j$ such that $\re_i \in \rs_j$. $\re_i$ is not covered, $\rs_j$ is not ranked the second in any votes in $P'$. Therefore, $P'$ does not change the $\re_i$'s score on $\rs_j$, which is $-q$. 
    \item $\rs_j$ such that $\re_i \not\in \rs_j$. we know that $P$ brings a score of $2q$ for $\re_i$ towards $\rs_j$. Since $|P'| = \frac{q}{3(\ell -1)} \le \frac{q}{3}$, the total score for $\re_i$ towards $\rs_j$ will not be lower than $\frac53q$. 
    \item $\re_h$ with $h\neq i$. $\re_i$ and $\re_h$ is tied in $P$. Therefore, the total score for $\re_i$ towards $\re_h$ will not be lower than $-\frac13 q$. 
    \item $\tg$ and $\wn_i$. the total score for $\re_i$ towards $\tg$ is exactly $q + 1$, and towards $\wn$ is at least $\frac{5}{3}q$
\end{itemize}
Therefore, the min-score of $\re_i$ is $-q$, which is higher than $\tg$. Therefore, $\tg$ cannot be the winner, which is a contradiction. Consequently, $S^*$ is an exact 3-cover of $\res$, and RXC3 is a YES instance. 
\end{proof}

\begin{thm}
\label{thm:maximinu}
    For any constant $L \ge 2$, \wpp-$\overline{\maximin}_{L}$ is NP-complete.
\end{thm}
The proof also follows the top-$\ell$ case by replacing $\ell$ with $L$. 

\section{Copeland}
\begin{thm}
\label{thm:copeland}
    For any constant $\ell\ge 2$ and any $\alpha \in [0, 1]$, \wpp-$\cpd_{\ell}^{\alpha}$ is NP-complete.  
\end{thm}

\begin{proof}
   The membership of NP is held by running the vote and checking the winner. For the hardness, we give a reduction from RXC3. We will apply Lemma~\ref{lem:McGarvey} to construct most of our profile yet use other structures when we need an odd number of weights. We first present the construction for $\alpha < 1$ (more precisely, $\alpha < \frac{q-3}{q}$, which converges to 1 as $q$ increases) and show how to modify the construction for $\alpha =1$. We assume that $q$ can be divided by $6(k-1)$. 

    \noindent\textbf{Candidates.} There are $2q+ \frac{q}{2}+3$ candidates: $\res\cup \rss\cup\{\tg, b\}\cup W$, where $W = \{ \wn_1, \wn_2, \cdots, \wn_{\frac{q}{2} + 1}\}$. W.l.o.g, we assume $\succ \re_1\succ \re_2\succ \cdots \succ \re_q \succ \tg\ \succ \wn_1\succ\cdots\succ\wn_{\frac{q}{2}+1}$ in tie-breaking. 

\noindent\textbf{Absent votes.} $t = \frac{q}{3(\ell-1)}$. 

Here we also give an intuition for the construction. The only edges that can be flipped by $P'$ are $\rs_j\to \tg$ and $\re_i \to \rs_j$ for $\re_i \in \rs_j$. We set the weights so that $\tg$ needs to win every $\rs_j$ to become the winner, which requires every vote in $P'$ to include $\tg$. On the other hand, every $\re_i$ needs to be tied with or be beaten by some $\rs_j \ni \re_i$ to make $\tg$ the winner. Therefore, the rest $\frac{q}{3}$ positions in $P'$ will be taken by $\rs_j$ that forms an exact 3-cover.

\noindent\textbf{Known votes.} 
We first construct the following edges with Lemma~\ref{lem:McGarvey}.
\begin{itemize}
    \item For each $\re_{i_1}$ and $\re_{i_2}$, if $i_2 - i_1 < \frac{q}{2} \mod q$, then there is an edge $\re_{i_1} \to \re_{i_2}$ of $4q$ weights; if $i_2 - i_1 > \frac{q}{2} \mod q$, then there is an edge $\re_{i_2} \to \re_{i_1}$ of weight $4q$. 
    \item For each $\re_{i_1}$ and $\re_{i_2}$ such that $i_2 - i_1 = \frac{q}{2} \mod q$, suppose $i_1 < i_2$. Then there are three edges: $\re_{i_1} \to \re_{i_2}$, $\re_{i_2} \to b$, and $b\to \re_{i_1}$, each of which of weight $4q$. 
    \item For each $\rs_{j_1}$ and $\rs_{j_2}$, if $j_2 - j_1 < \frac{q}{2} \mod q$, then there is an edge $\rs_{j_1} \to \rs_{j_2}$ of $4q$ weights; if $j_2 - j_1 > \frac{q}{2} \mod q$, then there is an edge $\rs_{j_2} \to \rs_{j_1}$ of weight $4q$.
    \item For each $\rs_{j_1}$ and $\rs_{j_2}$ such that $j_2 - j_1 = \frac{q}{2} \mod q$, suppose $j_1 < j_2$. Then there are three edges: $\rs_{j_1} \to \rs_{j_2}$, $\rs_{j_2} \to b$, and $\rs_{j_1} \to b$, each of which of weight $4q$. 
    \item For each $\wn_i$ and $\rs_j$, there is an edge for $\rs_j \to \wn_i$ of weight $4q$. 
    \item For each $\re_i$ and $\wn_j$, there is an edge for $\wn_j\to \re_i$ of weight $4q$. 
    \item For each $\re_i$ and $\rs_j$ such that $\re_i \not\in \rs_j$, there is an edge for each $\re_i \to rs_j$ of weight $4q$. 
    \item For each $\re_i$, there is an edge $\re_i\to \tg$ of weight $4q$. 
    \item For each $\wn_i$, there is an edge $\tg\to \wn_i$ of weight $4q$ and an edge $b \to \wn_i$ of weight $4q$.
    \item For each $\rs_j$, there is an edge $\rs_j \to \tg$ of weight $\frac{q}{3(\ell-1)} - 2$ .
    \item There is an edge of $b\to \tg$ with weights $4q$. 
\end{itemize}
Then we add up the following votes: 
\begin{itemize}
    \item For each $\re_i \in \rs_j$, there is one vote for $[\rs_j \succ \re_i \succ \wn_1\succ\cdots \succ \wn_{\ell-2}$.
    \item There are two votes for $[\tg \succ \wn_1\succ\cdots\succ \wn_{\ell-1}]$. 
\end{itemize}

The Copeland score for each candidate in $P$ is as follows: 
\begin{itemize}
    \item $\tg$: $\frac{q}{2}+1$
    \item $\re_i$: $q + \frac{q}{2} + 1$
    \item $\rs_j$: at most $q + 4$
    \item $\wn_i$: at most $q + \frac{q}{2}$. 
    \item $b$: $q + 2$.
\end{itemize}

Moreover, the weight of $\re_i \to \rs_j$ for each $\re_i\in \rs_j$ is $1$, and the weight for $rs_j \to \tg$ for each $\rs_j$ is $\frac{q}{3(\ell -1)} - 1$.

The weighted majority graph (with weights of key edges) of $P$ is in Figure~\ref{fig:cpd}.

\begin{figure}
    \centering
    \includegraphics[width = \columnwidth]{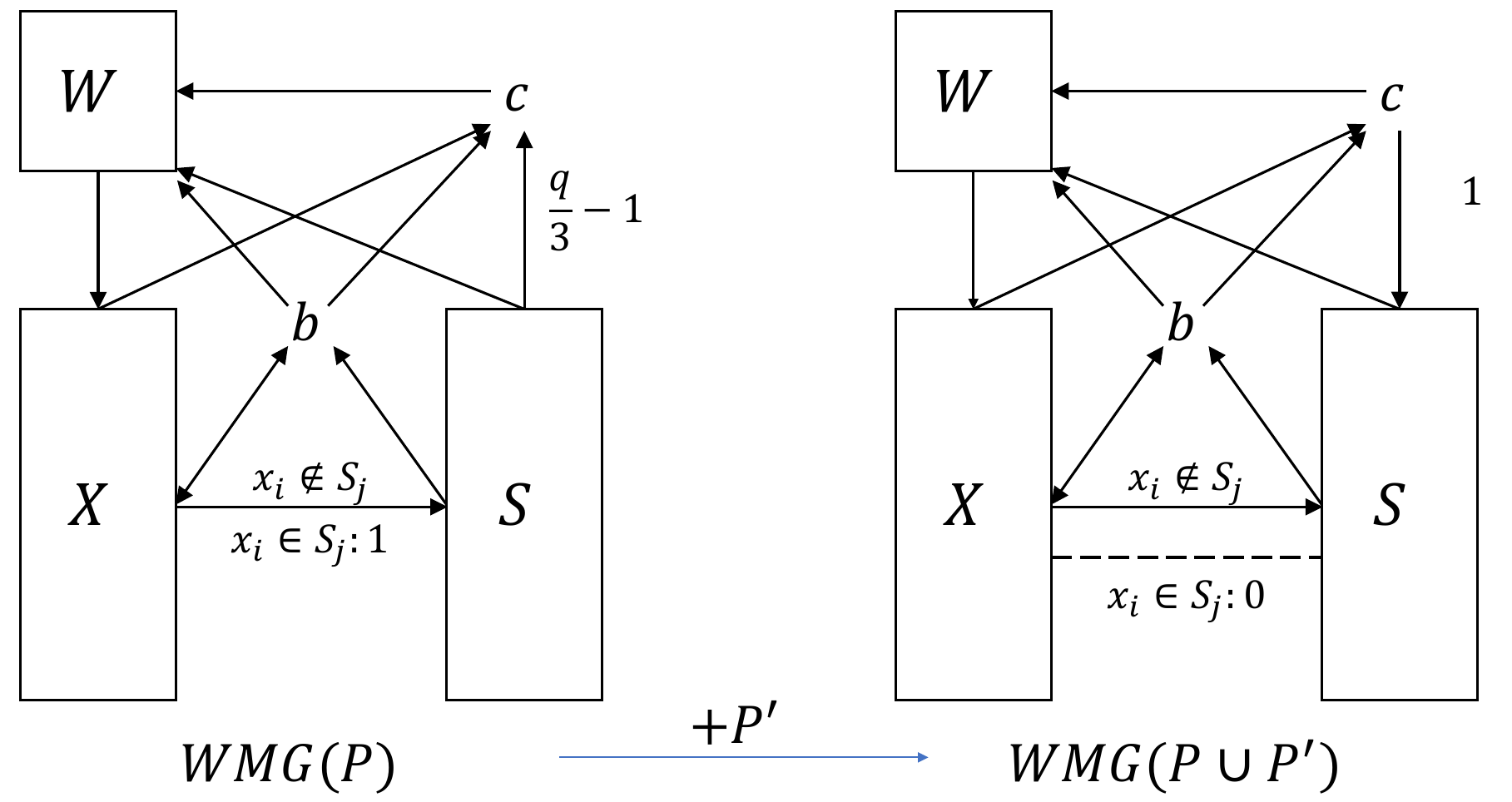}
    \caption{The WMG for Copeland.}
    \label{fig:cpd}
\end{figure}

\paragraph{Suppose RXC3 is a YES instance.} Suppose $\rss^*$ is a exact 3-cover of $\res$. We construct $P'$ as follows. $\tg$ is ranked the top and followed by $\ell -1$ of $\rs_j$ in each vote. The set of all the $\rs_j$ ranked second to the $\ell$-th in $P'$ is exactly $\rss^*$. Now consider the changes in the weighted majority graph. 
\begin{itemize}
    \item $\tg$ wins every $\rs_j$. 
    \item Each $\re_i$ wins all but one $\rs_j$ such that $\re_i \in \rs_j$ and $\rs_j \in \rss^*$, and $\re_i$ is tied with this $\rs_j$.  
    \item The directions of all other edges are unchanged. 
\end{itemize}
Then in $P\cup P'$, the Copeland score is as follows:
\begin{itemize}
    \item $\tg$: $q + \frac{q}{2}+1$
    \item $\re_i$: $(q-1) + \alpha + \frac{q}{2} + 1 = q + \frac{q}{2} + \alpha$
    \item $\rs_j$: at most $q + 3$
    \item $\wn_i$: at most $q + \frac{q}{2}$. 
    \item $b$: $q + 2$.
\end{itemize}
Therefore, $\tg$ becomes the winner. 

\paragraph{Suppose \wpp-$\cpd_{\ell}^{\alpha}$is a YES instance.} Suppose $P'$ is a profile that makes $\tg$ the winner. We give the proof in the following steps. 

Firstly, an important observation is that the edges in the graph whose directions can be changed by $P'$ are the edges between $\tg$ and $\rs_j$, and the edges between $\re_i \in \rs_j$. The edges inside $W$ can also be changed but will not affect the result. 

\noindent\textbf{Step 1.} All $\frac{q}{3}$ votes in $P'$ must contain $\tg$. If this is not the case, then in $P\cup P'$, $\tg$ will at best be tied with every $\rs_j$ in the head-to-head competition and get a Copeland score of at most $\alpha\cdot q + \frac{q}{2}+1$. On the other hand, the Copeland score of an $\re_i$ is at least $q-3 + \frac{q}{2} + 1$. Therefore, for every $\alpha < \frac{q-3}{q}$, $\tg$ cannot be the winner, which is a contradiction. 

\noindent\textbf{Step 2.} The set of all $\rs_j$ ranked in $P'$ is an exact 3-cover of $\res$. Let $\rss^*$ be the multiset of all $\rs_j$ in $P'$. Suppose $\rss^*$ is not an exact 3-cover of $\res$. Since there are at most $\frac{q}{3}$ of such $\rs_j$, there must exist an $\re_i^*$ not covered by $\rss^*$. Then the Copeland score of $\re_i^*$ will be $q + 1 + \frac{q}{2}$, which equals $\tg$'s score of $q + \frac{q}{2} + 1$. Then according to the tie-breaking rule, $\tg$ cannot be the winner, which is a contradiction. 

Therefore, $\rss^*$ is an exact 3-cover of $\res$, and RXC3 is a YES instance. 

\paragraph{$\alpha = 1$.}
For the $\alpha = 1$ case, we follow similar reasoning yet make a few changes in $P$: $\re_i$ is tied with all $\rs_j$ such that $\re_i \in \rs_j$, and the weight of each $\rs_j \to \tg$ becomes $\frac{q}{3(\ell - 1)}$. 
After adding up $P'$ to the profile, $\tg$ will be tied with every $\rs_j$, and $\re_i$ will be beaten by exactly one $\rs_j$. Then the score of $\tg$ is $q + \frac{q}{2} + 1$, each $\re_i$ is $q + \frac{q}{2}$, and $\tg$ becomes the winner. Otherwise, if some $\re_i$ is not covered, its score remains $q + \frac{q}{2} + 1$ and beats $\tg$. The rest of the reasoning remains the same as the $\alpha < 1$ case. 
\end{proof}

Similar to STV and Maximin, this proof also applies to the up-to-$L$ setting by replacing $\ell$ with $L$.  

\begin{thm}
    \label{thm:copelandu}
    For any constant $L \ge 2$ and any $\alpha \in [0, 1]$, \wpp-$\overline{\cpd}_{L}^{\alpha}$ is NP-complete. 
\end{thm}

\section{Positional Scoring Rules}
For (positional) scoring rules $\vs_{\ell} = (\scr_1, \scr_2, \cdots, \scr_{\ell})$. We show a special case when $\scr_2 = \cdots = \scr_{\ell}$, then the problem is in $P$. 

\begin{thm}
\label{thm:psr}
    For any $\ell\ge 2$ and $\vs_{\ell}$ such that $\scr_2 = \cdots = \scr_{\ell}$, \wpp-$\vs_{\ell}$ can be determined in polynomial time. 
\end{thm}

\begin{proof}
We convert the problem into a maximum flow problem. Given an instance of \wpp-$\vs_{\ell}$, we construct the following network flow. An illustration of the flow is in Figure~\ref{fig:psr}.

\noindent\textbf{Nodes}: $\{\bar{s}, \bar{t}\} \cup T \cup TM \cup M$. There are $(\ell-1)t + (t+1)(m-1) +2$ nodes.
\begin{itemize}
    \item $\bar{s}$ is the source node, and $\bar{t}$ is the sink node. 
    \item $T$ contains $(\ell-1)t$ nodes. For each absent vote $v$, there are $\ell-1$ nodes representing the second, third, $\cdots$, $\ell$-th position in the vote. The node representing $d$-th position of vote $v$ is denoted as $v_d$. 
    \item $TM$ contains $t(m-1)$ nodes. For each absent vote $v$ and each candidate $a\neq \tg$, there is a node $(v, a)$ in $TM$. 
    \item $M$ contains $m-1$ nodes, each representing a candidate other than $\tg$. 
\end{itemize}

\noindent\textbf{Edges}: $E_1\cup E_2\cup E_3\cup E_4$. 
\begin{itemize}
    \item For each node $v_d\in T$, there is an edge $s\to v_d$ with capacity 1. 
    \item For each vote $v$, each position $d$, and each candidate $a$, there is an edge $v_d \to (v, a)$ with capacity 1. 
    \item For each vote $v$ and each candidate $a\neq \tg$, there is an edge $(v, a)\to a$ with capacity $1$. 
    \item Let $\scr_2 = \cdots = \scr_{\ell} = A$, and let $\vs(P, a)$ be the score of candidate $a$ from profile $P$. For each node $a\in M$, there is an edge $a\to \bar{t} $ with capacity $\lfloor \frac{\vs(P, \tg) - \vs(P, a) + t\cdot \scr_1}{A} \rfloor$. (We assume all such capacity is non-negative. If this is not the case, the \wpp-$\vs_{\ell}$ is directly a NO instance.)
\end{itemize}

\begin{figure}
    \centering
    \includegraphics[width = \columnwidth]{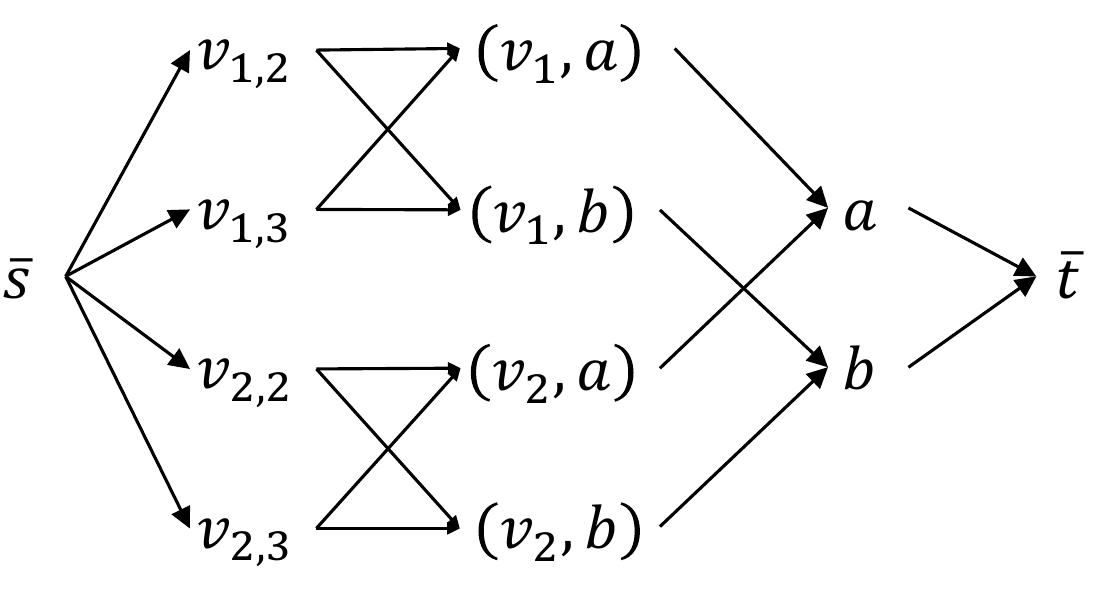}
    \caption{An illustration of the network when $M =\{a, b, c\}$, $t =2$, and $\ell =3$. }
    \label{fig:psr}
\end{figure}

Here we give an interpretation of the network. For a flow $f$:
\begin{itemize}
    \item $f(v_d \to (v, a))=1$ stands for that $a$ denotes that candidate $a$ is ranked at $d$-th in vote $v$. 
    \item $f((v,a)\to a) = 1$ stands for that $a$ appears in the top-$\ell$ ranking in $v$. The capacity ensures that every candidate appears at most once in $v$. 
    \item $A\cdot f(a\to \bar{t})$ stands for the total score candidate $a$ gets from $P'$. The capacity ensures that the total score of $a$ does not exceed the total score of $\tg$ from $P\cup P'$. 
\end{itemize}

We show that \wpp-$\vs_{\ell}$ is a YES instance if and only if the maximum flow of the network is $(\ell-1)t$. 

\paragraph{Suppose the maximum flow is $(\ell-1)t$.} Suppose $f$ is a maximum flow. Without loss of generality, we could assume that $f$ is an integer flow (by the Ford-Fulkerson algorithm~\citep{ford_fulkerson_1956}). We construct a solution $P'$ that makes $\tg$ a winner.

Firstly, $\tg$ is ranked the first in all $t$ votes in $P'$. 
For all edge $e = (v, a) \to a$ such that $f(e) = 1$, there must exist exactly one $2\le l\le k$ such that $f(v_d \to (v, a)) = 1$.  Then for the vote $v\in P'$, we assign $d$-th position of $v$ to be candidate $a$. Since the maximum flow is $(\ell-1)t$, all $v_d \in T$ has a candidate $a$ such that $f(v_d\to (v, a))=1$, and all $(\ell-1)t$ positions is assigned to some candidate $a\neq \tg$. Moreover, for each vote, $v$, each candidate $a$ appears at most once guaranteed by the capacity of $(v, a) \to a$. Therefore, each $v$ is a valid ranking. Finally, since the capacity of each $a\to \bar{t}$ is not reached, no candidate gets a higher score than $\tg$ (assume $\tg$ has the highest priority in tie-breaking), and $\tg$ will be the winner. 

\paragraph{Suppose \wpp-$\vs_{\ell}$ is a YES instance.} Suppose $P'$ is a solution that makes $\tg$ a winner. Without loss of generality, we could assume that $\tg$ is ranked the top at all votes $v\in P'$. (By ranking $\tg$ the top while keeping other order unchanged, the score of $\tg$ is non-decreasing and the score of any other candidate is non-increasing). Then we can construct a flow $f$ with flow amount $(\ell-1)t$. For each vote $v$, if candidate $a$ is ranked $d$-th in $v$, we set $f(v_d\to (v, a)) =1$ and $f((v,a) \to a) = 1$. Then we set the flow of all edges in $E_1$ and $E_4$ according to the conservation of flow. With similar reasoning, we can verify that the flow amount is $(\ell-1)t$.  
\end{proof}

Theorem~\ref{thm:psr} directly indicates that \wpp-$\vs_{\ell}$ is in P for any $\vs_{\ell}$ when $\ell =2$. 
\begin{coro}
\label{coro:psr2}
    \wpp-$\vs_2$ can be determined in polynomial time for any $\vs_2$. 
\end{coro}

In the up-to-$L$ setting, determining whether a candidate $\tg$ can be a winner under the up-rounding scoring rule can be achieved by checking if $\tg$ becomes the winner when all votes in $P'$ are $[c]$, i.e. rank $c$ alone. 

\begin{prop}[\citep{Narodytska14:Computational}, {Proposition 1}] For any constant $L \ge 1$ and any scoring vector $\vs_{L}$, \wpp-${\overline{\vs}_{L\uparrow}}$ can be computed in polynomial time.
\end{prop}

For the down-rounding scoring rule, we show a similar result as the top-$\ell$ setting.

\begin{thm}
    \label{thm:psrud}
    For any constant $L\ge 2$ and $\vs_{L}$ such that $\scr_2 = \cdots = \scr_{L}$, \wpp-$\overline{\vs}_{L\downarrow}$ can be determined in polynomial time. 
\end{thm}
\begin{proof}[Proof Sketch]
    Firstly, if a solution $P'$ exists, we could assume without loss of generality that $P'$ contains only top-1 ranking and top-$L$ ranking, and $\tg$ is ranked the top in all the votes. If this is not the case, we could substitute all non-top-$L$ votes into $[c]$, and rank $\tg$ the top of all the top-$L$ votes while keeping the order of other candidates unchanged. In this way, the scoring of $\tg$ is strictly increasing, while the score of all other candidates is non-increasing. 

    Then we give the computation solution scheme. First, we enumerate the number of top-1 votes and top-$L$ votes in $P'$. The sum of two kinds of votes is $\un$. Therefore, there are in total $t + 1$ cases. For each case, we set all top-1 votes to be $[c]$, and construct a maximum flow instance as in the proof of Theorem~\ref{thm:psr}. If there is some case where the maximum flow is above the threshold, then we output YES. Otherwise, when all cases the maximum flow is below its threshold, we output NO. 
\end{proof}

\section{Conclusion and Future Work}
We investigate the computational complexity of determining winners with absent votes when the votes are top-truncated. We have shown that the problem is in P when the number of candidates or the quantity of absent votes is bounded. In the unbounded cases, we show that the problem is NP-complete for STV, Maximin, and Copeland. We also give a special case of scoring rules where the problem can be computed in polynomial time. Winner determination with absent votes is closely related to the classic coalitional manipulation problem in social choice, yet previous results on full rankings do not directly extend to top-truncated settings. 

A question that remains open in our paper is the complexity of \wpp{} for general positional scoring rules. In the full-ranking setting, the complexity of coalitional manipulation is regarded as a challenging task. \citet{Xia10:Scheduling} shows the hardness under an artificially constructed scoring vector. \citet{Betzler11:Unweighted} shows the hardness under Borda. Other interesting directions include studying the possible winner problem under top-truncated votes, or more powerful manipulators that may truncate/extend certain votes. 





\bibliographystyle{ACM-Reference-Format} 
\bibliography{ref_new,references}


\clearpage
\onecolumn
\appendix
\section{Proof for STV when $\ell=2$ and $\ell = 3$}
 If these two cases, we only need to ``truncate'' the rankings of voters to top-2 or top-3 rank. The only influence of this change is that the votes that were supposed to be transferred to $\wn$ or $\tg$ now vanish. However, this will not change the result. Once one of $\wn$ and $\tg$ is eliminated, the other will get its votes and achieve a score of at least $24q -1$. Other alternatives, on the other hand, cannot have score higher than $12q + 8i+q/3 \le 20q + q/3 $. Therefore, truncating the ranks will not change the result, and the proof still holds. 

\begin{proof}[Proof of $\ell=2$]
For {\bf \boldmath $\ell = 2$}, we have the following construction: 

\begin{itemize}
\item {\boldmath $\prf_1$}: There are $12q$ votes of $[\tg\succ \wn]$.
\item {\boldmath $\prf_2$}: There are $12q-1$ votes of $[\wn\succ \tg]$.
\item {\boldmath $\prf_3$}: There are $10q+ 2q/3$ votes of $[\da_0\succ \wn]$.
\item {\boldmath $\prf_4$}: For every $i\in \{1,\ldots,q\}$, there are  $12q-2$ votes of $[\da_i\succ \wn]$.
\item {\boldmath $\prf_5^1$}: For every $i\in \{1,\ldots,q\}$, there are $6q+4i-6$ votes of $[\ba_i\succ\bb_i]$; and {\boldmath $\prf_5^2$}: for every $i\in \{1,\ldots,q\}$ and  every $j$ such that $\re_j\in \rs_i$, there are two votes of $[\ba_i\succ \da_j]$. 
\item {\boldmath $\prf_6^1$}:  For every $i\in \{1,\ldots,q\}$, there are $6q+4i-2$ votes of $[\bb_i\succ \ba_i]$; and {\boldmath $\prf_6^2$}: for every $i\in \{1,\ldots,q\}$, there are two votes of $[\bb_i\succ \da_0]$.
\end{itemize}
First still, {\bf the ultimate winner will be either $\tg$ or $\wn$}. Once one of $\tg$ or $\wn$ is eliminated in some round, the remaining other will gets its votes of $\prf_1$ or $\prf_2$: if $c$ is eliminated first, $\wn$ will get $12q$ votes from $\prf_1$; and if $\wn$ is eliminated first $\wn$ will get $12q-1$ votes from $\prf_2$. Therefore, the remaining one will have score of at least $24q-1$. On the other hand, all other alternatives cannot have such a high score:
\begin{itemize}
    \item $\ba_i$ and $\bb_i$ only gets votes from each other and from $\prft$, and cannot have score more than $12q + 8i+ q/3 \le 20q + q/3$. 
    \item $\da_0$ can get two votes from each $\bb_i$ and votes from $\prft$, and cannot have score more than $10q + 2/3q + 2q + q/3 = 13q$. 
    \item $\da_i$ can get votes from $\ba_i$ and from $\prft$. Since each $\re_j$ is in exactly three $\rs_i$, $\da_i$ can get two votes from exactly three $\ba_i$. Therefore, $\da_i$ cannot have score more than $12q - 2 + 6 + q/3 = 12q + q/3 +4$. 
\end{itemize}
Therefore, no other alternative can exceed the score of $24q-1$ at any time, and cannot be the winner. 
Now we show that \wpp-$STV_{\ell}$ is a YES instance if and only if RXC3 is a YES instance. 

{\bf Suppose RXC3 is a YES instance.} Let $\idx$ be the index set such that $|\idx| = q/3$ and $\bigcup_{i\in\idx} \rs_i = \res$. Then we construct $\prft$ as follows: for each $i\in \idx$, there is one vote of $[\bb_i\succ\ba_i]$. Then in the first $q$ round of voting, for each $i\le q$, if $i\in\idx$, $\ba_i$ is eliminated; otherwise $\bb_i$ is eliminated. If $\ba_i$ is eliminated, $6q+4i-6$ of its vote transfer to $\bb_i$, and for every $j$ such that $\re_j\in\rs_i$, $\da_j$ gets two of its vote. If $\bb_i$ is eliminated, $6q+4i-2$ of its vote transfer to $\ba_i$, and two transfer to $\da_0$. Therefore, in the beginning of $q+1$ round, the plurality scores of the remaining alternatives are as in the following table.
\renewcommand{\arraystretch}{1.5}
\begin{center}
\begin{tabular}{|@{\ }c|@{\ }c@{\ } @{\ }|@{\ }c@{\ }|@{\ }c@{\ }|@{\ }c@{\ }|@{\ }c@{\ }|}
\hline  
Rd. & $\wn$ & $\tg$ & $\ba_i$ or $\bb_i$& $\da_0$& $ \da_i$\\
\hline $q+1$& $12q-1$ & $12q$ &  \begin{tabular}{@{}c@{}}$12q+8i-1$ or  \\$12q+8i-5$\end{tabular}& $12q $ & $12q $\\
\hline
\end{tabular}
\end{center}
\renewcommand{\arraystretch}{1}
Therefore, $\wn$ is eliminated in round $q+1$, whose votes transfer to $\tg$, and it's not hard to verify that $\tg$ will be the ultimate winner.

{\bf Suppose \wpp-$STV_{\ell}$ is a YES instance.} We prove that RXC3 is a YES instance in the following steps. 

{\bf Step 1.} In the first $q$ rounds, exactly one of $\ba_i$ and $\bb_i$ is eliminated for all $i\le q$. Firstly, the initial score of $\ba_i$ and $\bb_i$ is at most $6q+4i+q/3 \le 10q+q/3$, while the score of other alternatives is at least $10q + 2/3q$. Moreover, $\ba_i$ and $\bb_i$ will not get any vote transfer from alternatives expect for each other. On the other hand, once one of $\ba_i$ and $\bb_i$ is eliminated, the other gets the transferred votes and has score more than $12q$. Therefore, in the first $q$ round, in each round either $\ba_i$ or $\bb_i$ is eliminated for a distinct $i$. 

{\bf Step 2.} Let $\idx = \{i:\ \ba_i\text{ is eliminated in the first $q$ rounds.}\}$. Then $\idx$ must be the index set of an RXC3 solution.  

Suppose it is not the case. Firstly, $|\idx| \le q/3$. For each $i \in \idx$, $\bb_i$ needs at least one vote for $P'$ to win $\ba_i$ in the round that $\ba_i$ is eliminated. And once $\bb_i$ is not eliminated, this vote follows $\bb_i$ to the round $q+1$ and cannot contribute to another $\bb_i$'s winning. Therefore, since $t = q/3$, there are at most $q/3$ of $\bb_i$ that beats $\ba_i$. 

Now suppose $|\idx| = q/3$ and $\idx$ is not a solution. Then there exists some $\re_j \in \res$ that is not covered. In this case, for all $i$ such that $\re_j\in\rs_i$, $\bb_i$ is eliminated. Therefore, $\da_j$ does not get any vote transfer from the first $q$ rounds. On the other hand, all $q/3$ votes contribute to some $\bb_i$ in the $q+1$ round and cannot contribute to $\da_j$. Therefore, $\da_j$ has a score of $12q - 2$ in the $q+1$ round. Therefore, one of such $\da_j$ is eliminated in this round, and its votes are transferred to $\wn$. Then $\wn$ has a score of at least $24q -3$ which exceeds $\tg$ all the time. Therefore, $\tg$ cannot be the winner.

Finally, suppose $|\idx| < q/3$. Then there are at least $q - 3|\idx|$ of $\re_j$ not covered, and all the corresponding $\da_j$ do not get transferred in the first $q$ rounds. If they shall not be eliminated in the $q+1$ round, there needs to be at least 1 vote for each of them, which is $q - 3|\idx|$ votes from $P'$. However, since all $|\idx|$ votes contribute to some $\bb_i$, there are only $q/3 - |\idx|$ and cannot cover all $\da_j$. Therefore, one $\da_j$ is eliminated in this round, and its votes are transferred to $\wn$. Then $\wn$ has a score of at least $24q -3$ which exceeds $\tg$ all the time. Therefore, $\tg$ cannot be the winner.

Therefore, once \wpp-$STV_{\ell}$ is a YES instance, the index set of eliminated $\ba_i$ in the first $q$ rounds forms the index set of a solution to the RXC3. Therefore, RXC3 is a YES instance. 
\end{proof}

\begin{proof}[Proof of $\ell = 3$]
For {\bf \boldmath $\ell = 3$}, we have the following construction: 

\begin{itemize}
\item {\boldmath $\prf_1$}: There are $12q$ votes of $[\tg\succ \wn]$.
\item {\boldmath $\prf_2$}: There are $12q-1$ votes of $[\wn\succ \tg]$.
\item {\boldmath $\prf_3$}: There are $10q+ 2q/3$ votes of $[\da_0\succ \wn\succ \tg]$.
\item {\boldmath $\prf_4$}: For every $i\in \{1,\ldots,q\}$, there are  $12q-2$ votes of $[\da_i\succ \wn\succ\tg]$.
\item {\boldmath $\prf_5^1$}: For every $i\in \{1,\ldots,q\}$, there are $6q+4i-6$ votes of $[\ba_i\succ\bb_i\succ\wn]$; and {\boldmath $\prf_5^2$}: for every $i\in \{1,\ldots,q\}$ and  every $j$ such that $\re_j\in \rs_i$, there are two votes of $[\ba_i\succ \da_j\succ\wn]$. 
\item {\boldmath $\prf_6^1$}:  For every $i\in \{1,\ldots,q\}$, there are $6q+4i-2$ votes of $[\bb_i\succ \ba_i\succ\wn]$; and {\boldmath $\prf_6^2$}: for every $i\in \{1,\ldots,q\}$, there are two votes of $[\bb_i\succ \da_0\succ\wn]$.
\end{itemize}
First still, {\bf the ultimate winner will be either $\tg$ or $\wn$}. Once one of $\tg$ or $\wn$ is eliminated in some round, the remaining other will gets its votes of $\prf_1$ or $\prf_2$: if $c$ is eliminated first, $\wn$ will get $12q$ votes from $\prf_1$; and if $\wn$ is eliminated first $\wn$ will get $12q-1$ votes from $\prf_2$. Therefore, the remaining one will have score of at least $24q-1$. On the other hand, all other alternatives cannot have such a high score:
\begin{itemize}
    \item $\ba_i$ and $\bb_i$ only gets votes from each other and from $\prft$, and cannot have score more than $12q + 8i+ q/3 \le 20q + q/3$. 
    \item $\da_0$ can get two votes from each $\bb_i$ and votes from $\prft$, and cannot have score more than $10q + 2/3q + 2q + q/3 = 13q$. 
    \item $\da_i$ can get votes from $\ba_i$ and from $\prft$. Since each $\re_j$ is in exactly three $\rs_i$, $\da_i$ can get two votes from exactly three $\ba_i$. Therefore, $\da_i$ cannot have score more than $12q - 2 + 6 + q/3 = 12q + q/3 +4$. 
\end{itemize}
Therefore, no other alternative can exceed the score of $24q-1$ at any time, and cannot be the winner. 
Now we show that \wpp-$STV_{\ell}$ is a YES instance if and only if RXC3 is a YES instance. 

{\bf Suppose RXC3 is a YES instance.} Let $\idx$ be the index set such that $|\idx| = q/3$ and $\bigcup_{i\in\idx} \rs_i = \res$. Then we construct $\prft$ as follows: for each $i\in \idx$, there is one vote of $[\bb_i\succ\ba_i\succ\tg]$. Then in the first $q$ round of voting, for each $i\le q$, if $i\in\idx$, $\ba_i$ is eliminated; otherwise $\bb_i$ is eliminated. If $\ba_i$ is eliminated, $6q+4i-6$ of its vote transfer to $\bb_i$, and for every $j$ such that $\re_j\in\rs_i$, $\da_j$ gets two of its vote. If $\bb_i$ is eliminated, $6q+4i-2$ of its vote transfer to $\ba_i$, and two transfer to $\da_0$. Therefore, in the beginning of $q+1$ round, the plurality scores of the remaining alternatives are as in the following table.
\renewcommand{\arraystretch}{1.5}
\begin{center}
\begin{tabular}{|@{\ }c|@{\ }c@{\ } @{\ }|@{\ }c@{\ }|@{\ }c@{\ }|@{\ }c@{\ }|@{\ }c@{\ }|}
\hline  
Rd. & $\wn$ & $\tg$ & $\ba_i$ or $\bb_i$& $\da_0$& $ \da_i$\\
\hline $q+1$& $12q-1$ & $12q$ &  \begin{tabular}{@{}c@{}}$12q+8i-1$ or  \\$12q+8i-5$\end{tabular}& $12q $ & $12q $\\
\hline
\end{tabular}
\end{center}
\renewcommand{\arraystretch}{1}
Therefore, $\wn$ is eliminated in round $q+1$, whose votes transfer to $\tg$, and it's not hard to verify that $\tg$ will be the ultimate winner.

{\bf Suppose \wpp-$STV_{\ell}$ is a YES instance.} We prove that RXC3 is a YES instance in the following steps. 

{\bf Step 1.} In the first $q$ rounds, exactly one of $\ba_i$ and $\bb_i$ is eliminated for all $i\le q$. Firstly, the initial score of $\ba_i$ and $\bb_i$ is at most $6q+4i+q/3 \le 10q+q/3$, while the score of other alternatives is at least $10q + 2/3q$. Moreover, $\ba_i$ and $\bb_i$ will not get any vote transfer from alternatives expect for each other. On the other hand, once one of $\ba_i$ and $\bb_i$ is eliminated, the other gets the transferred votes and has score more than $12q$. Therefore, in the first $q$ round, in each round either $\ba_i$ or $\bb_i$ is eliminated for a distinct $i$. 

{\bf Step 2.} Let $\idx = \{i:\ \ba_i\text{ is eliminated in the first $q$ rounds.}\}$. Then $\idx$ must be the index set of an RXC3 solution.  

Suppose it is not the case. Firstly, $|\idx| \le q/3$. For each $i \in \idx$, $\bb_i$ needs at least one vote for $P'$ to win $\ba_i$ in the round that $\ba_i$ is eliminated. And once $\bb_i$ is not eliminated, this vote follows $\bb_i$ to the round $q+1$ and cannot contribute to another $\bb_i$'s winning. Therefore, since $t = q/3$, there are at most $q/3$ of $\bb_i$ that beats $\ba_i$. 

Now suppose $|\idx| = q/3$ and $\idx$ is not a solution. Then there exists some $\re_j \in \res$ that is not covered. In this case, for all $i$ such that $\re_j\in\rs_i$, $\bb_i$ is eliminated. Therefore, $\da_j$ does not get any vote transfer from the first $q$ rounds. On the other hand, all $q/3$ votes contribute to some $\bb_i$ in the $q+1$ round and cannot contribute to $\da_j$. Therefore, $\da_j$ has a score of $12q - 2$ in the $q+1$ round. Therefore, one of such $\da_j$ is eliminated in this round, and its votes are transferred to $\wn$. Then $\wn$ has a score of at least $24q -3$ which exceeds $\tg$ all the time. Therefore, $\tg$ cannot be the winner.

Finally, suppose $|\idx| < q/3$. Then there are at least $q - 3|\idx|$ of $\re_j$ not covered, and all the corresponding $\da_j$ do not get transferred in the first $q$ rounds. If they shall not be eliminated in the $q+1$ round, there needs to be at least 1 vote for each of them, which is $q - 3|\idx|$ votes from $P'$. However, since all $|\idx|$ votes contribute to some $\bb_i$, there are only $q/3 - |\idx|$ and cannot cover all $\da_j$. Therefore, one $\da_j$ is eliminated in this round, and its votes are transferred to $\wn$. Then $\wn$ has a score of at least $24q -3$ which exceeds $\tg$ all the time. Therefore, $\tg$ cannot be the winner. 

Therefore, once \wpp-$STV_{\ell}$ is a YES instance, the index set of eliminated $\ba_i$ in the first $q$ rounds forms the index set of a solution to the RXC3. Therefore, RXC3 is a YES instance. 
\end{proof}

\end{document}

%% file: macro.tex
\theoremstyle{definition}
\newtheorem{thm}{Theorem}
\newenvironment{thmbis}[1]
  {\renewcommand{\thethm}{\ref{#1}}%
   \addtocounter{thm}{-1}%
   \begin{thm}}
  {\end{thm}}
\newtheorem{dfn}{Definition}
\newenvironment{dfnbis}[1]
  {\renewcommand{\thedfn}{\ref{#1}}%
   \addtocounter{dfn}{-1}%
   \begin{dfn}}
  {\end{dfn}}
\newtheorem{conj}{Conjecture}
\newtheorem{lem}{Lemma}
\newenvironment{lembis}[1]
  {\renewcommand{\thelem}{\ref{#1}}%
   \addtocounter{lem}{-1}%
   \begin{lem}}
  {\end{lem}}
\newtheorem{ex}{Example}
\newtheorem{Alg}{Algorithm}
\newtheorem{prob}{Problem}
\newtheorem{question}{Question}
\newtheorem{prop}{Proposition}
\newtheorem{coro}{Corollary}
\newenvironment{prof}{\noindent{\em Proof.}\rm }{\hfill $\Box$ }
\newenvironment{sketch}{\noindent{\em Proof sketch.}\rm }{\hfill $\Box$ }
\newtheorem{cond}{Condition}
\newtheorem{claim}{Claim}
\newtheorem{mes}{Message}
\newtheorem{view}{Viewpoint}
\newtheorem{calculation}{Calculation}
\newtheorem{obs}{Observation}
\newtheorem{RQ}{Research Question}
\newcounter{newct}
\newcommand{\rqu}{%
        \stepcounter{newct}%
        Research Question~\thenewct.}

\newcommand\qishen[1]{{\color{blue} \footnote{\color{blue}Qishen: #1}} }

\newcommand\amelie[1]{{\color{blue} \footnote{\color{blue}Amelie: #1}} }

\newcommand{\blue}[1]{\textcolor{blue}{#1}}

\newcommand{\cd}{m} 
\newcommand{\cds}{\mathcal{A}} 
\newcommand{\cdl}{\mathcal{L(A)}} 
\newcommand{\rk}{R} 
\newcommand{\vt}{n} 
\newcommand{\un}{t} 
\newcommand{\tg}{c} 
\newcommand{\wn}{w} 
\newcommand{\topk}{k} 
\newcommand{\prf}{P} 
\newcommand{\prfn}{P}
\newcommand{\prft}{P'}
\newcommand{\rl}{r} 
\newcommand{\stv}{\text{STV}}
\newcommand{\maximin}{\text{Maximin}}
\newcommand{\stvk}{\text{STV}_k}
\newcommand{\istv}{\overline{\text{STV}}}
\newcommand{\copeland}{\text{Cd}_\alpha}
\newcommand{\cpd}{\text{Cd}}
\newcommand{\lex}{{\sc Lex}}
\newcommand{\wpp}{\text{WAV}}
\newcommand{\wppr}[1]{\text{\sc WP}\text{-}{#1}}
\newcommand{\rtr}[1]{\text{\sc RT}\text{-}{#1}}
\newcommand{\rs}{S}
\newcommand{\rss}{\mathcal{S}}
\newcommand{\re}{x}
\newcommand{\res}{X}

\newcommand{\ba}{b}
\newcommand{\bb}{\overline{b}}
\newcommand{\da}{d}

\newcommand{\idx}{I}
\newcommand{\jdx}{J}

\newcommand{\calS}{\mathcal{S}}
\newcommand{\wmg}{\text{WMG}}

\newcommand{\vs}{\vec{s}}
\newcommand{\scr}{a}